\documentclass[12pt]{article}

\usepackage{customtemplate}
\newcommand{\qbinom}[2]{\genfrac{[}{]}{0pt}{}{#1}{#2}_q}

\newcommand{\F}{\mathbb{F}}

\newcommand{\X}{\mathbf{X}}

\def\Fq{{\mathbb F}_q}
\def\a{{\alpha}}

\def\AA{{\mathbb A}}

\def\l{{\ell}}
\def\lp{{\l^\prime}}

\def\PP{{\mathbb P}}
\def\clm{{C^\AA(\ell,m)}}

\DeclareMathOperator{\ev}{Ev}

\DeclareMathOperator{\wt}{wt}

\theoremstyle{definition}

\numberwithin{theorem}{section}

\title{Revisiting the Weight Spectrum of the Affine Grassmann Code $C^{\AA}(2,m)$}
\author{Rohit Yadav\footnote{Email: rohityadavau1998@gmail.com\\
Department of Mathematics, Indian Institute of Technology, Jammu, Jammu--181221, India.}}

\date{}

\begin{document}
\maketitle
\begin{abstract}
Affine Grassmann codes, introduced by Beelen, Ghorpade and H{\o}holdt, are linear codes over $\F_q$ obtained by evaluating linear combinations of minors of a generic matrix. The weight spectrum of the affine Grassmann code $C^{\AA}(2,m)$ was determined by Pi\~nero and Singh via a case analysis on the rank of an associated alternating matrix. We give an independent and more streamlined derivation, valid for all $m\ge4$ and every prime power $q$, in which each codeword is written in a compact matrix form and its Hamming weight is expressed through a single closed formula involving two explicit affine subspaces and their intersection.
\end{abstract}



\section{Introduction}\label{sec:intro}

Evaluation codes are linear codes obtained by evaluating a fixed vector space of functions at every point of a fixed set of points. Projective systems \cite{TVN2007} provide a useful general framework for studying such codes. A classical example is the Reed--Muller family, whose codewords are obtained by evaluating polynomials of bounded degree at every point of an affine space. Affine Grassmann codes generalize this construction by evaluating linear combinations of minors of a generic matrix at every point of the corresponding affine space over $\Fq$.

Let $\ell$ and $\ell'$ be positive integers with $\ell\le\ell'$, put $m:=\ell+\ell'$ and $\delta:=\ell\ell'$, and let $\AA^\delta(\Fq)$ denote the $\Fq$-rational points of the affine open cell of the Grassmannian $G_{\ell,m}$ defined by the nonvanishing of a chosen Pl\"ucker coordinate. In the standard local coordinates on this cell, a point of $\AA^\delta(\Fq)$ is naturally identified with an $\ell\times\ell'$ matrix over $\Fq$, since the $\delta$ coordinates are precisely the entries of such a matrix. This identification underlies the entire construction and is used throughout the paper. Under the Pl\"ucker embedding, this cell is itself a projective system $\AA^\delta(\Fq)\subseteq\PP(\Fq)^{\binom{m}{\ell}-1}$ (see \cite{BGH2010}*{Section VII}), and the resulting code, denoted $C^{\AA}(\ell,m)$, is the affine Grassmann code associated with it. When $\ell=1$, the only minors of the generic $1\times\ell'$ matrix are the constant $1$ and its $\ell'$ entries, so $C^{\AA}(1,m)$ is exactly the $q$-ary first-order Reed--Muller code $\mathrm{RM}_q(1,\ell')$. For $\ell\ge2$, the higher-order minors provide genuinely new information, and it is in this sense that affine Grassmann codes generalize first-order Reed--Muller codes (cf. \cite{BGH2010}*{Sec.~1}). Affine Grassmann codes were first studied by Beelen, Ghorpade and H{\o}holdt \cite{BGH2010}, who showed that $C^{\AA}(\ell,m)$ is an $[n,k,d]_q$ code with
\begin{equation}\label{eq:parameters}
n=q^{\delta},\qquad k=\binom{m}{\ell},\qquad\text{and}\qquad
d=q^{\delta-\ell^2}\prod_{i=0}^{\ell-1}\bigl(q^{\ell}-q^{i}\bigr),
\end{equation}
characterized its minimum weight codewords, and showed that its automorphism group is quite large. Its dual $C^{\AA}(\ell,m)^{\perp}$ was studied in a subsequent article \cite{BGT2012}, where it was described as an evaluation code of certain functions on the set of $\ell\times\ell'$ matrices.

The automorphism group of $C^{\AA}(\ell,m)$ was later determined completely by Ghorpade and Kaipa \cite{GK2013}, and Datta and Ghorpade \cite{DG2015} computed some of its initial and terminal generalized Hamming weights. A majority logic decoder for these codes over nonbinary fields, building on earlier work on Grassmann codes \cite{BPP2021} and Schubert codes \cite{S2022}, was proposed in \cite{PPR2026}.

Among the invariants of a code, its weight spectrum, which records the weight of every codeword together with the number of codewords of each weight, is particularly informative. It determines the weight spectrum of the dual code via the MacWilliams identities, and it governs the probability of a decoding error on a memoryless channel \cite{MS1977}. It also refines the minimum distance, since it records how many codewords lie at every distance, rather than only the smallest distance. For affine Grassmann codes, this invariant is known only when $\ell=2$. Pi\~nero and Singh \cite{PS2019} showed that a codeword of $C^{\AA}(2,m)$ corresponds to an alternating matrix whose rank governs its weight, and determined the spectrum through a case analysis on that rank. The first case with $\ell\ge3$, namely $C^{\AA}(3,6)$, was subsequently determined by the present author and Prasant Singh \cite{SinghYadav2026}. Beyond that, the weight spectrum of $C^{\AA}(\ell,m)$ remains open.

In this paper, we revisit the case $\ell=2$ and give an independent proof of its weight spectrum. We replace the case analysis of \cite{PS2019} with a single closed formula expressing the Hamming weight of a codeword in terms of the sizes of two explicit affine subspaces of $\Fq^{\ell'}$ and their intersection. In this formula, the zero codeword and the nonzero constant codewords appear as degenerate instances, rather than as cases requiring separate treatment, and every step works uniformly over every prime power $q$, without requiring a separate argument in characteristic two. The uniform treatment is also useful beyond the case $\ell=2$. Specialising one row of the generic matrix carries a codeword of $C^{\AA}(\ell,m)$ to a codeword of $C^{\AA}(\ell-1,m-1)$, and the closed formula proved here is the key ingredient in that reduction when $\ell=3$ (see \cite{SinghYadav2026}).

This article is organized as follows. Section~\ref{sec:weight-dist} recalls the general construction of affine Grassmann codes and fixes the notation used throughout, and then contains our main results. Theorem~\ref{Theo_Weight_2m} gives the unified weight formula for codewords of $C^{\AA}(2,m)$, Theorem~\ref{thm:count-2m} assembles the complete weight distribution, and we conclude with a worked example.

\section{Weight Enumeration of \texorpdfstring{$C^{\mathbb A}(2,m)$}{CA(2,m)}}
\label{sec:weight-dist}

We first recall the general construction of the affine Grassmann code $\clm$ (see \cite{BGH2010} for details). Let $\X=(X_{ij})$ denote the generic $\ell\times\ell'$ matrix, whose entries $X_{ij}$ are independent variables, and let $\mathcal F(\ell,m)$ be the $\Fq$-linear span of all minors of $\X$, where the $0\times0$ minor is taken to be the constant $1$. Then $\mathcal F(\ell,m)$ is a $\binom{m}{\ell}$-dimensional subspace of $\Fq[\X]$. As noted in the introduction, a point of $\AA^\delta(\Fq)$ is naturally identified with an $\ell\times\ell'$ matrix over $\Fq$, obtained by assigning a value in $\Fq$ to each of the $\delta$ entries $X_{ij}$, and evaluating a minor $f\in\mathcal F(\ell,m)$ at such a point simply means substituting these values into the polynomial $f$. In other words, the minors of the generic matrix $\X$ are evaluated, in the classical sense, at every $\ell\times\ell'$ matrix over $\Fq$.

Choose an enumeration $P_1,\dots,P_{q^\delta}$ of $\AA^\delta(\Fq)$, that is, of the $\ell\times\ell'$ matrices over $\Fq$. The evaluation map
\[
\ev:\mathcal F(\ell,m)\longrightarrow\Fq^{q^\delta},\qquad
f\longmapsto\bigl(f(P_1),\dots,f(P_{q^\delta})\bigr)
\]
is $\Fq$-linear and injective \cite{BGH2010}, and its image is the code $\clm$ described in the introduction. For $f\in\mathcal F(\ell,m)$ we write $c_f:=\ev(f)$, so the Hamming weight of $c_f$ is simply the number of matrices $P\in\AA^\delta(\Fq)$ at which $f(P)$ is nonzero.

The weight distribution of $C^{\AA}(2,m)$ was first found by Pi\~nero and Singh \cite{PS2019}. We give an independent derivation below, which reduces the weight of a codeword to a simple point count involving two affine subspaces of $\Fq^{\lp}$ and their intersection.

From now on, we fix $\ell=2\le \lp$ and $m=2+\lp$, so $\delta=2\lp$. By \eqref{eq:parameters}, $C^{\AA}(2,m)$ then has length $q^{2\lp}$ and dimension $\binom{m}{2}$. When $\ell=2$, the matrix $\X$ has only two rows, so its minors reduce to the entries $X_{ij}$ themselves and the $2\times2$ minors $X_{1j_1}X_{2j_2}-X_{1j_2}X_{2j_1}$, and every $f\in\mathcal F(2,m)$ therefore has a unique expression
\[
f(X)= c_0+ \sum_{j=1}^{\lp} a_{1j}X_{1j} + \sum_{j=1}^{\lp} a_{2j}X_{2j}
+ \sum_{1\le j_1 < j_2\le \lp} b_{j_1j_2}
\bigl(X_{1j_1}X_{2j_2}-X_{1j_2}X_{2j_1}\bigr),
\]
with $c_0,a_{ij},b_{j_1j_2}\in\Fq$.

To these coefficients, we attach the vectors
\[
\a_1:=(a_{11},\ldots,a_{1\lp}),\qquad
\a_2:=(a_{21},\ldots,a_{2\lp})^{t},
\]
and the $\lp\times\lp$ matrix $B=(B_{jk})$ with
\[
B_{j_2j_1}:=b_{j_1j_2},\qquad B_{j_1j_2}:=-b_{j_1j_2}\quad (j_1<j_2),
\qquad B_{jj}:=0 .
\]
Thus, $B$ is \emph{alternating}, meaning that $B^{t}=-B$ and all diagonal entries are zero. (In characteristic two, $B^{t}=-B$ alone would not force the diagonal to vanish. Here it does so automatically, since the entries of $B$ come directly from the $2\times2$ minors.)

Let $X_1 := (X_{11},\ldots,X_{1\lp})^{t}$ and $X_2 := (X_{21},\ldots,X_{2\lp})$ denote the transpose of the first row and the second row of $\X$, respectively. A direct expansion of $X_2BX_1=\sum_{j,k}X_{2j}B_{jk}X_{1k}$ then gives
\begin{equation}\label{Rep_f}
f(X_1,X_2)= c_0 + \a_1 X_1 + X_2 \a_2 + X_2 B X_1 .
\end{equation}
Conversely, for every choice of $c_0\in\Fq$, $\a_1,\a_2\in\Fq^{\lp}$, and an alternating $B$, there exists a unique $f$ giving rise to these parameters, as is confirmed by the parameter count
\[
1+2\lp+\binom{\lp}{2}=\binom{2+\lp}{2}=\binom{m}{2}=\dim\mathcal F(2,m).
\]

Since a point of $\AA^\delta(\Fq)$ is, for $\ell=2$, a matrix with exactly two rows, evaluating $f$ at such a point amounts to assigning values to $X_1$ and $X_2$, and we may therefore identify the evaluation points of $\AA^{\delta}(\Fq)$ with pairs $(u,v)\in\Fq^{\lp}\times\Fq^{\lp}$, where $u$ is the transpose of the first row and $v$ is the second row of the matrix. By \eqref{Rep_f},
\begin{equation}\label{eq:fuv}
f(u,v)= \bigl(c_0+\a_1u\bigr) + v\bigl(\a_2+Bu\bigr).
\end{equation}
Thus, for fixed $u$, the function $v\mapsto f(u,v)$ is affine in $v$, with linear part $v\mapsto v(\a_2+Bu)$ and constant term $c_0+\a_1u$. This observation is the key to the weight computation.

We now distinguish the values of $u$ for which the linear part vanishes from those for which the constant term vanishes.
\[
U:=\{u\in\Fq^{\lp}: \a_2+Bu=0\},\qquad
V:=\{u\in\Fq^{\lp}: c_0+\a_1u=0\}.
\]
Both sets are affine subspaces of $\Fq^{\lp}$, although either one may be empty. Here $B$ gives a linear map $u\mapsto Bu$ from $\Fq^{\lp}$ to itself, while $\a_1$ gives a linear functional $u\mapsto \a_1u$ on $\Fq^{\lp}$. Put
\[
U_0:=\ker B,\qquad V_0:=\ker \a_1,\qquad N:=\#U,\qquad M:=\#(U\cap V).
\]
Since $B$ is alternating, its rank is even, say $\operatorname{rank}B=2r$ with $0\le 2r\le \lp$, so $\dim U_0=\lp-2r$. The next lemma determines the possible values of $N$ and $M$.

\begin{lemma}\label{lem:NM}
With the notation above, the following hold.
\begin{enumerate}
\item[(1)] We have $U\neq\emptyset$ if and only if $-\a_2\in\operatorname{Im}B$. In
that case $U=u_2+U_0$ for any $u_2$ with $Bu_2=-\a_2$, so
\[
N=\begin{cases} 0,& -\a_2\notin\operatorname{Im}B,\\ q^{\lp-2r},&\text{otherwise.}\end{cases}
\]
\item[(2)] If $\a_1=0$, then $M=0$ when $c_0\neq0$, and $M=N$ when $c_0=0$.
\item[(3)] Let $\a_1\neq0$ and $U\neq\emptyset$, and fix $u_2$ as in (1). Then
$U\cap V\neq\emptyset$ if and only if $c_0+\a_1u_2=0$, and in that case
$M=q^{\dim(U_0\cap V_0)}$, where
\[
\dim(U_0\cap V_0)=\begin{cases}\lp-2r, & U_0\subseteq V_0,\\ \lp-2r-1,& U_0\nsubseteq V_0.\end{cases}
\]
Hence $M$ takes values in $\{0,\;q^{\lp-2r-1},\;q^{\lp-2r}\}$, and the last value occurs only when
$U_0\subseteq V_0$.
\end{enumerate}
\end{lemma}

\begin{proof}
(1) The condition $\a_2+Bu=0$ says $Bu=-\a_2$. This is solvable if and only if $-\a_2$ lies in the image of the linear map $u\mapsto Bu$. The solution set of an inhomogeneous linear system, when nonempty, is a coset of the kernel. So $U=u_2+U_0$ and $N=\#U_0=q^{\lp-2r}$.

(2) If $\a_1=0$, then $V=\emptyset$ for $c_0\neq0$, and $V=\Fq^{\lp}$ for $c_0=0$. Both claims are immediate.

(3) Let $\a_1\neq0$, so $V=u_1+V_0$ with $\a_1u_1=-c_0$ and $\dim V_0=\lp-1$. A point $p\in U\cap V$ can be written as $p=u_2+x=u_1+y$ with $x\in U_0$ and $y\in V_0$. Thus, $U\cap V\neq\emptyset$ if and only if $u_1-u_2\in U_0+V_0$. Since $V_0$ is a hyperplane, there are two possibilities for the relative position of $U_0$ and $V_0$. If $U_0\subseteq V_0$, then $U_0+V_0=V_0$. If $U_0\nsubseteq V_0$, then $U_0+V_0=\Fq^{\lp}$, so the condition holds automatically. In the first case, $u_1-u_2\in V_0$ is equivalent to $\a_1(u_1-u_2)=0$, or equivalently, $c_0+\a_1u_2=0$. In the second case, $U_0\nsubseteq V_0$ implies $\a_1(U_0)=\Fq$, so $U\cap V$ is always nonempty.

In either case, choose $p\in U\cap V$. Then $U=p+U_0$ and $V=p+V_0$, because a coset is unchanged when its representative is replaced by any of its elements. Hence $U\cap V=p+(U_0\cap V_0)$ and $M=q^{\dim(U_0\cap V_0)}$. Finally, $U_0\cap V_0$ is the intersection of $U_0$ with the hyperplane $V_0$. So its dimension is $\dim U_0$ if $U_0\subseteq V_0$, and $\dim U_0-1$ otherwise.
\end{proof}

\begin{theorem}\label{Theo_Weight_2m}
Let $f=c_0+\a_1X_1+X_2\a_2+X_2BX_1\in\mathcal F(2,m)$ and let $c_f=\ev(f)$. With $N=\#U$ and $M=\#(U\cap V)$ as above,
\begin{equation}\label{eq:wt-unified}
\wt(c_f)=q^{2\lp}-q^{2\lp-1}+N\,q^{\lp-1}-M\,q^{\lp} .
\end{equation}
Explicitly, if $\operatorname{rank}B=2r$, then $\wt(c_f)$ takes one of the three values
\[
\wt(c_f)=
\begin{cases}
q^{2\lp}-q^{2\lp-1}, & N=0,\ \text{or } M=q^{\lp-2r-1},\\[2pt]
q^{2\lp}-q^{2\lp-1}+q^{2\lp-2r-1}, & N\neq0,\ M=0,\\[2pt]
q^{2\lp}-q^{2\lp-1}-q^{2\lp-2r}+q^{2\lp-2r-1}, & N\neq0,\ M=q^{\lp-2r}.
\end{cases}
\]
\end{theorem}

\begin{proof}
We split the sum defining the weight according to whether $u\in U$:
\begin{equation}\label{Weight_eq}
\wt(c_f)=\sum_{u\in \Fq^{\lp} \setminus U}\#\{v \in \Fq^{\lp}: f(u,v)\neq0\}
        +\sum_{u\in \Fq^{\lp}\cap U}\#\{v \in \Fq^{\lp}: f(u,v)\neq0\}.
\end{equation}
Let $u\notin U$. Then $\a_2+Bu\neq0$, so by \eqref{eq:fuv} the map $v\mapsto f(u,v)$ is a nonconstant affine function of $v$. Its zero set is a coset of the hyperplane $\{v \in \Fq^{\lp} : v(\a_2+Bu)=0\}$, so it has $q^{\lp-1}$ elements. Hence
\[
\#\{v \in \Fq^{\lp}: f(u,v)\neq0\}=q^{\lp}-q^{\lp-1}\qquad (u\notin U),
\]
and the contribution from $u\notin U$ is $(q^{\lp}-N)(q^{\lp}-q^{\lp-1})$.

Now let $u\in U$. Then $\a_2+Bu=0$, so \eqref{eq:fuv} reduces to the constant $f(u,v)=c_0+\a_1u$, which does not depend on $v$. Therefore, $\#\{v \in \Fq^{\lp}:f(u,v)\neq0\}$ equals $q^{\lp}$ if $u\notin V$, and $0$ if $u\in V$. Hence the contribution from $u\in U$ is $q^{\lp}\cdot\#(U\setminus V)=q^{\lp}(N-M)$.

Adding the two contributions,
\[
\wt(c_f)=(q^{\lp}-N)(q^{\lp}-q^{\lp-1})+q^{\lp}(N-M)
=q^{2\lp}-q^{2\lp-1}+Nq^{\lp-1}-Mq^{\lp},
\]
which is \eqref{eq:wt-unified}. Substituting the possible values of $N$ and $M$ from Lemma~\ref{lem:NM} gives the explicit list. If $N=0$, then $M=0$, so the formula gives $q^{2\lp}-q^{2\lp-1}$. If $N=q^{\lp-2r}$, then $M\in\{0,q^{\lp-2r-1},q^{\lp-2r}\}$. The three substitutions give $q^{2\lp}-q^{2\lp-1}+q^{2\lp-2r-1}$, then $q^{2\lp}-q^{2\lp-1}+q^{2\lp-2r-1}-q^{2\lp-2r-1}=q^{2\lp}-q^{2\lp-1}$, and finally $q^{2\lp}-q^{2\lp-1}+q^{2\lp-2r-1}-q^{2\lp-2r}$.
\end{proof}

\begin{remark}
Formula \eqref{eq:wt-unified} also covers the two degenerate codewords, so they need no separate discussion. If $f=0$, then $B=0$, $\a_2=0$, $\a_1=0$ and $c_0=0$, so $U=V=\Fq^{\lp}$ and $N=M=q^{\lp}$, which gives $\wt(c_f)=0$. If $f=c_0$ is a nonzero constant codeword, then $U=\Fq^{\lp}$ and $V=\emptyset$, so $N=q^{\lp}$ and $M=0$, which gives $\wt(c_f)=q^{2\lp}$. Both cases are included in the $r=0$ part of the explicit list in Theorem~\ref{Theo_Weight_2m}.
\end{remark}

For $0\le 2r\le\lp$, set
\[
s_r:=q^{r(r-1)}\prod_{i=1}^{r}\bigl(q^{2i-1}-1\bigr),
\qquad
n_{2r}:=\qbinom{\lp}{2r}s_r,
\qquad
k_{2r}:=\qbinom{\lp-1}{\lp-2r}s_r ,
\]
where $\qbinom{a}{b}$ denotes the Gaussian binomial coefficient. The next lemma identifies these three quantities. Recall that a subspace $W\subseteq\Fq^{\lp}$ of dimension $\lp-2r$ is the kernel of some alternating matrix of rank $2r$.

\begin{lemma}\label{lem:altcount}
The following hold.
\begin{enumerate}
\item[(1)] The quantity $s_r$ counts the nondegenerate alternating
$2r\times2r$ matrices over $\Fq$. Equivalently, for a fixed subspace $W \in \Fq^{\lp}$ of dimension $\lp-2r$, it is the number of alternating $\lp\times\lp$ matrices $B$ of rank $2r$
with $\ker B=W$.
\item[(2)] The quantity $n_{2r}$ counts the alternating $\lp\times\lp$ matrices of
rank $2r$, and $\sum_{r}n_{2r}=q^{\binom{\lp}{2}}$.
\item[(3)] For a fixed nonzero $\a_1$, the quantity $k_{2r}$ counts the
alternating $\lp\times\lp$ matrices $B$ of rank $2r$ with $\ker B\subseteq\ker\a_1$.
\end{enumerate}
\end{lemma}

\begin{proof}
(1) An alternating matrix $B$ of rank $2r$ with kernel $W$ corresponds to a nondegenerate alternating form on the $2r$-dimensional quotient $\Fq^{\lp}/W$. The number of nondegenerate alternating $2r\times2r$ matrices is given by the classical formula $q^{r(r-1)}\prod_{i=1}^{r}(q^{2i-1}-1)$.

(2) To choose an alternating matrix $B$ of rank $2r$, we first choose its kernel, a subspace of dimension $\lp-2r$, and then choose a nondegenerate alternating form on the quotient. There are $\qbinom{\lp}{\lp-2r}=\qbinom{\lp}{2r}$ choices for the kernel, so $n_{2r}=\qbinom{\lp}{2r}s_r$ by (1). Every alternating matrix has even rank, so summing over all admissible values of $r$ counts all alternating matrices.

(3) The kernel $\ker\a_1$ is a hyperplane, so it has dimension $\lp-1$. By the same argument, the number of admissible kernels is $\qbinom{\lp-1}{\lp-2r}$.
\end{proof}

For $0\le 2r\le \lp$, write
\[
W_r^{-}:=q^{2\lp}-q^{2\lp-1}-q^{2\lp-2r}+q^{2\lp-2r-1},
\qquad
W_r^{+}:=q^{2\lp}-q^{2\lp-1}+q^{2\lp-2r-1},
\]
and $W_0^{\ast}:=q^{2\lp}-q^{2\lp-1}$, and let $A_t$ denote the number of codewords of
$C^{\AA}(2,m)$ of weight $t$. Combining Theorem~\ref{Theo_Weight_2m} with
Lemma~\ref{lem:altcount} now gives the complete weight distribution of $C^{\AA}(2,m)$.

\begin{theorem}\label{thm:count-2m}
The weights $W_r^{-}$, $W_r^{+}$, and $W_0^{\ast}$, for $0\le 2r\le\lp$, are pairwise
distinct, and together they exhaust the weights of $C^{\AA}(2,m)$. Moreover,
\[
A_{W_r^{-}}=q^{2r}\Bigl(n_{2r}+(q^{\lp}-1)\,k_{2r}\Bigr), \qquad
A_{W_r^{+}}=(q-1)\,A_{W_r^{-}}, \qquad
A_{W_0^{\ast}}=q^{\binom{m}{2}}-\sum_{r=0}^{\lfloor \lp/2\rfloor} q\,A_{W_r^{-}} .
\]
In particular, $r=0$ gives $W_0^{-}=0$ with $A_0=1$, and $W_0^{+}=q^{2\lp}$ with
$A_{q^{2\lp}}=q-1$.
\end{theorem}

\begin{proof}
We have $W_r^{-}<W_0^{\ast}<W_r^{+}$ for every $r$. Also $W_r^{\pm}$ determines $r$, because the exponents $2\lp-2r$ and $2\lp-2r-1$ are strictly decreasing in $r$. So the listed weights are pairwise distinct, and by Theorem~\ref{Theo_Weight_2m} every codeword has one of these weights. It remains to count the codewords of each weight.

Fix $r$ and an alternating $B$ of rank $2r$. By Lemma~\ref{lem:NM}(1), we have $U\neq\emptyset$ for exactly $q^{2r}$ of the $q^{\lp}$ vectors $\a_2$, namely those with $-\a_2\in\operatorname{Im}B$. For the remaining $q^{\lp}-q^{2r}$ vectors, we have $U=\emptyset$. When $N=0$, the weight is necessarily $W_0^{\ast}$. We therefore restrict to the $q^{2r}$ vectors $\a_2$ for which $U\neq\emptyset$ and count the codewords of weights $W_r^{-}$ and $W_r^{+}$. By Theorem~\ref{Theo_Weight_2m}, these are the codewords with $M=q^{\lp-2r}$ and with $M=0$ respectively. The remaining possibility, $M=q^{\lp-2r-1}$, also gives the weight $W_0^{\ast}$.

First, let $\a_1=0$. By Lemma~\ref{lem:NM}(2), we get $M=N=q^{\lp-2r}$ when $c_0=0$, and $M=0$ when $c_0\neq0$. Choose one of the $n_{2r}$ alternating matrices $B$ of rank $2r$ and one of the $q^{2r}$ admissible vectors $\a_2$. Taking $c_0=0$ gives one codeword of weight $W_r^{-}$, while taking $c_0\neq0$ gives $q-1$ codewords of weight $W_r^{+}$. Thus, this case contributes $q^{2r}n_{2r}$ codewords of weight $W_r^{-}$ and $q^{2r}(q-1)n_{2r}$ codewords of weight $W_r^{+}$.

Now let $\a_1\neq0$. By Lemma~\ref{lem:NM}(3), the case $M=q^{\lp-2r}$ requires $U_0\subseteq V_0$. When $U_0\nsubseteq V_0$, we always have $M=q^{\lp-2r-1}$, which gives the weight $W_0^{\ast}$. So assume $U_0=\ker B\subseteq V_0=\ker\a_1$. By Lemma~\ref{lem:altcount}(3), there are $k_{2r}$ such $B$ for each of the $q^{\lp}-1$ choices of $\a_1$. Fix such a pair, fix one of the $q^{2r}$ admissible vectors $\a_2$, and let $u_2$ satisfy $Bu_2=-\a_2$. Since $U_0\subseteq V_0$, the value $\a_1u_2$ does not depend on the choice of $u_2$. So, by Lemma~\ref{lem:NM}(3), the condition $U\cap V\neq\emptyset$ reads $c_0=-\a_1u_2$. This condition holds for exactly one of the $q$ possible values of $c_0$ and fails for the remaining $q-1$ values. The unique value giving $M=q^{\lp-2r}$ produces the weight $W_r^{-}$, whereas the remaining $q-1$ values give $M=0$ and hence the weight $W_r^{+}$. Thus, this case contributes $(q^{\lp}-1)q^{2r}k_{2r}$ codewords of weight $W_r^{-}$ and $(q^{\lp}-1)q^{2r}(q-1)k_{2r}$ codewords of weight $W_r^{+}$.

Adding the two cases gives $A_{W_r^{-}}=q^{2r}\bigl(n_{2r}+(q^{\lp}-1)k_{2r}\bigr)$. In both cases, the number of codewords of weight $W_r^{+}$ is $(q-1)$ times that of weight $W_r^{-}$, since $c_0$ has $q-1$ possible nonzero values in the former case and a unique value in the latter. Hence $A_{W_r^{+}}=(q-1)A_{W_r^{-}}$. Finally, $\dim\mathcal F(2,m)=\binom{m}{2}$, so there are $q^{\binom m2}$ codewords in all. Then $A_{W_0^{\ast}}$ is what remains after removing, for each $r$, the $A_{W_r^{-}}+A_{W_r^{+}}=q\,A_{W_r^{-}}$ codewords of weight $W_r^{\pm}$.

For $r=0$, we have $B=0$, so $n_0=1$ and $k_0=\qbinom{\lp-1}{\lp}=0$. Hence $A_{W_0^{-}}=1$ and $A_{W_0^{+}}=q-1$, while $W_0^{-}=0$ and $W_0^{+}=q^{2\lp}$. These are the zero codeword and the $q-1$ nonzero constants.
\end{proof}

\begin{example}
Let $q=3$ and $\lp=3$, so $m=5$. Then $C^{\AA}(2,5)$ has length $3^{6}=729$ and $3^{10}=59049$ codewords. In this case, $r\in\{0,1\}$. For $r=0$, we get the weights $0$ and $729$, with multiplicities $1$ and $2$, respectively. For $r=1$, where all Gaussian binomial coefficients are evaluated at $q=3$, we get $s_1=2$, $n_2=\qbinom{3}{2}\cdot 2=13\cdot2=26$ and $k_2=\qbinom{2}{1}\cdot2=4\cdot2=8$. Hence $A_{W_1^{-}}=9\,(26+26\cdot8)=2106$ and $A_{W_1^{+}}=2\cdot2106=4212$, with $W_1^{-}=432$ and $W_1^{+}=513$. Finally, $A_{486}=59049-3(1+2106)=52728$. Thus, the five weights $0,432,486,513,729$, with multiplicities $1,2106,52728,4212,2$, account for all $59049$ codewords. 
\end{example} 

\begin{remark}
By Theorem~\ref{thm:count-2m}, $C^{\AA}(2,m)$ has exactly $2\lfloor m/2\rfloor+1$ distinct weights, compared with $\lfloor m/2\rfloor+1$ for the Grassmann code $C(2,m)$ itself \cite{Nogin1996}. Thus, passing to the affine cell gives approximately twice as many distinct weights. The weight formula \eqref{eq:wt-unified} obtained here is used in \cite{SinghYadav2026} to determine the weight spectrum of $C^{\AA}(3,6)$, the first affine Grassmann code with $\ell=3$. In that work, specialising one row of the generic matrix reduces a codeword of $C^{\AA}(\ell,m)$ to a codeword of $C^{\AA}(\ell-1,m-1)$, and it is precisely this formula that makes the reduction work.
\end{remark}

\subsection*{Acknowledgment}
The author would like to thank Prasant Singh for useful discussions on this problem.

\bibliographystyle{plain}
\bibliography{Version_2.bib}

\end{document}